\documentclass{iopart}

\makeatletter
\renewcommand{\tableofcontents}{\section*{\contentsname} \@starttoc{toc}}
\makeatother
\expandafter\let\csname equation*\endcsname\reax
\expandafter\let\csname endequation*\endcsname\relax

\usepackage{bbm, amsmath, amssymb, amsthm, bm, textcomp}
\usepackage[utf8]{inputenc}
\usepackage[T1]{fontenc}
\usepackage[english]{babel}
\usepackage{geometry,graphicx,enumitem}

\def\be{\begin{eqnarray}}

\def\ee{\end{eqnarray}}

\theoremstyle{definition}
\newtheorem{Def}{Definition}
\newtheorem{Exp}{Example}
\theoremstyle{plain}
\newtheorem{Pro}{Proposition}
\newtheorem{Lem}{Lemma}

\setitemize{leftmargin=*}

\begin{document}

\title{Certifiability criterion for large-scale quantum systems}

\author{F Fr\"owis$^1$, M van den Nest$^2$ and W D\"ur$^1$}
\address{$^1$ Institut f\"ur Theoretische Physik, Universit\"at  Innsbruck, Technikerstra\ss e 25, 6020 Innsbruck, Austria\\
$^2$ Max-Planck-Institut f\"ur Quantenoptik, Hans-Kopfermann-Straße 1, 85748 Garching, Germany.}
\ead{florian.froewis@unige.ch}
\date{\today}

\begin{abstract}
Can one certify the preparation of a coherent, many-body quantum state by measurements with bounded accuracy in the presence of noise and decoherence?
Here, we introduce a criterion to assess the fragility of large-scale quantum states which is based on the distinguishability of orthogonal states after the action of very small amounts of noise.
States which do not pass this criterion are called asymptotically incertifiable. We show that, if a coherent quantum state is asymptotically incertifiable, there exists an incoherent mixture (with entropy at least $\log 2$)  which is experimentally indistinguishable from the initial state.  
The Greenberger-Horne-Zeilinger states are examples of such asymptotically incertifiable states. More generally, we prove that any so-called macroscopic superposition state is asymptotically incertifiable. We also provide examples of quantum states that are experimentally indistinguishable from highly incoherent mixtures, i.e., with an almost-linear entropy in the number of qubits. Finally we show that all unique ground states of local gapped Hamiltonians (in any dimension) are certifiable.

\end{abstract}


\maketitle

\section{Introduction}
\label{sec:introduction}

The recent experimental advances in control of quantum systems has lead to the possibility to experimentally generate and manipulate quantum systems consisting of a large number of particles. As the system size increases, the detection of quantum coherences or entanglement in the produced states becomes a challenge. One obstacle is the exponential growth of the state space with the number of particles, requiring an exponential amount of measurements to fully characterise a given state. While this restriction can in principle be overcome by measuring only few characteristic properties of interest (e.g., using entanglement witnesses \cite{guhne_entanglement_2009}) a more fundamental challenge remains. It is almost impossible to fully shield a large system from its environment, so effects of noise and decoherence need to be considered. In a typical scenario where noise affects individual particles independently, this leads however to an exponential decrease of fidelity and other quantities of interest with system size, even for a tiny amount of noise per particle. Besides noise, one has to take into account that measurement devices only have finite precision. In addition, measurement statistics fundamentally limits the reachable accuracy with which any quantity can be determined. Hence, the following fundamental question remains: Is it in principle possible to confirm that a quantum state with certain desired properties is present---or has been prepared before the action of noise---if one deals with large systems?

Here, we introduce a fragility criterion to answer this question for many-body quantum states.
More precisely, we deal with the question whether one can certify by means of measurements with a bounded accuracy
that a coherent quantum state $|\psi\rangle$ (or a state close to it) was prepared {\em before} the action of noise or decoherence; or whether the measured signal is compatible with (and hence indistinguishable from) a situation where the initial state was \emph{orthogonal} to $|\psi\rangle$ and thus \textit{a priori} fully distinguishable from it. We consider a scenario where states are subjected to a very small amount of noise, where it is assumed that all particles are affected individually. The source of this noise can either be an imperfect preparation of the state or decoherence induced by interactions of the particles with an uncontrollable environment. If, after the action of the noise, two pure states $|\psi\rangle$ and $|\psi^{\perp}\rangle$ which were initially orthogonal become  indistinguishable by the given measurement device and a bounded number of measurements, then the states are said to be \emph{incertifiable}. This implies that preparation of either state cannot be confirmed, as the measured signal could come from either one of the two initial states. What is more, in this case it is even impossible to determine whether the initial state was a coherent quantum state at all, since (as we will show) the measured signal could equally well come from \emph{any incoherent mixture of $|\psi\rangle$ and $|\psi^{\perp}\rangle$}.  If, on the other hand, a state is found to be distinguishable from {\em all} orthogonal states after the action of the noise,  the state  is said to be certifiable. That is, the measurement device allows one in principle to distinguish between initial situations corresponding to the presence of a pure quantum state (or a state very close to it) and an arbitrary incoherent mixture of several orthogonal states. As a result, a criterion to assess the stability of large-scale coherent quantum states is obtained.

Such a criterion can indeed be formalised and we show that quantitative results can be obtained for large classes of states.)

\begin{itemize}
\item We find that all ground states of local, gapped Hamiltonians
  are certifiable. This includes all unentangled states (i.e.,
  products states), but also states that are highly entangled such as
  cluster states \cite{briegel_persistent_2001} (or graph states with
  a bounded local degree \cite{hein_entanglement_2005}) and Dicke
  states \cite{dicke_coherence_1954}.

\item  We show that Greenberger-Horne-Zeilinger (GHZ) states
  \cite{greenberger_going_1989} are incertifiable in the limit of
  large particle numbers. What is more, {\em all} states that
  correspond to so-called macroscopic superpositions (according to
  criteria specified in
   \cite{bjork_size_2004,korsbakken_measurement-based_2007,marquardt_measuring_2008,frowis_measures_2012})
  are proved to be incertifiable.

\item Interestingly, we also find that certain quantum states can also
  be indistinguishable from many orthogonal states in the presence of
  small amounts of noise. This implies that such states are, in the
  considered set-up, indistinguishable from incoherent mixtures with
  high entropy.
\end{itemize}

We remark that ideas which are to a certain extent related to our approach, have been discussed in
the context of decoherence theory and open system dynamics (see e.g.~\cite{zurek_pointer_1981,zurek_preferred_1993,zurek_decoherence_2003,joos_decoherence_2003,schlosshauer_decoherence:_2007,wismann_optimal_2012}). On the one hand, it is well known that coherences of certain large-scale system are unstable under certain noise processes,
which is e.g.~used in decoherence theory to explain the emergence of a classical world \cite{zurek_pointer_1981,zurek_preferred_1993,zurek_decoherence_2003,joos_decoherence_2003,schlosshauer_decoherence:_2007}. What we
add to this discussion is a precise formulation of some of the intuitive concepts used in this context,
together with a discussion of consequences to experimentally certify the creation of large-scale
quantum states. Finally, the approach to investigate trajectories of orthogonal state pairs under
noise was also taken in \cite{wismann_optimal_2012}, although there the aim was rather different as it was related to obtaining a measure for non-Markovianity of the noise process.

The paper is organised as follows. In section \ref{sec:criterion-stability}, we introduce and further motivate our criterion to asses the stability of coherent quantum states of large-scale systems. In Sec. \ref{sec:basic_examples} we discuss two elementary and illustrative examples: We show that all product states are certifiably, while GHZ states are incertifiable. In section \ref{sec:unique-ground-states}, we provide a general proof that all unique ground states of local, gapped Hamiltonians are certifiable. As an application, we show that Dicke states are certifiable. In section \ref{sec:macr-superp-are}, we show in contrast that all macroscopic superposition states are incertifiable, while we provide examples of states that are indistinguishable from many orthogonal states in section \ref{sec:inst-with-resp}. In section \ref{sec:stab-neighb-inst}, we present an extension of the proposed criteria to $\epsilon$-balls around initial pure states to correctly handle the neighbourhood of (in)certifiable states.  We summarise and conclude in section \ref{sec:conclusion-summary}. The appendix provides a glossary of some basic notions and proof details.

\section{Certifiability of quantum states under noise}
\label{sec:criterion-stability}

In this section, we define---as one of our main contributions---a criterion to asses the stability of coherent quantum states of large-scale systems. This will be accomplished in three steps. In section \ref{sec:motiv_defn} we introduce the notion of \emph{certifiability under noise}. In section \ref{sec:noise_model} we describe the specific noise model considered throughout this work. In section \ref{sec:asymptotic_certifiability} we describe our stability criterion (Definition \ref{def:asymp_cert}) which is related to the asymptotic behaviour of the certifiability measure.

The Hilbert space we consider in the following is $\mathcal{H}=\mathbbm{C}_d^{\otimes N}$ with $d,N \in \mathbbm{N}_{\geq 2}$. Normalised vectors represent pure states of $N$ $d$-level qudits. The set of density matrices $\rho$ on this Hilbert  space is denoted by $\mathcal{D}(\mathcal{H})$. We will concentrate on $d=2$, i.e., systems of qubits, however the methods and concepts we present are trivially extendable to arbitrary $d$.

\begin{figure}[htbp]
\centerline{\includegraphics[width=.8\columnwidth]{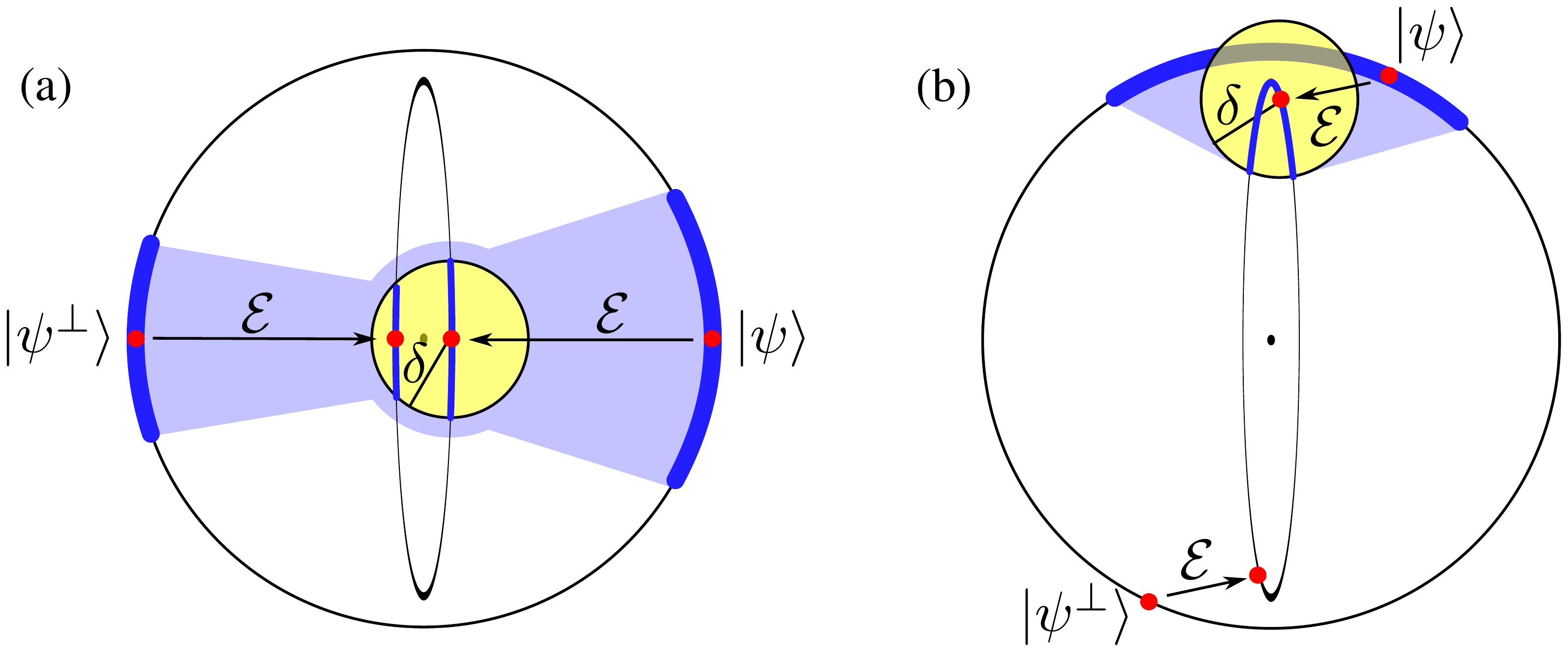}}
\caption[]{\label{fig:intuitivePicture} Schematic illustration of generalised Bloch sphere spanned by two orthogonal vectors $\left| \psi \right\rangle $ and $| \psi^{\bot} \rangle $.  The outer circles represent pure states. The action of $\mathcal{E}$ maps every pure state onto a mixed state situated on the inner ellipse.  The inner, yellow circle represent the measurement uncertainty induced by the apparatus. The thick, blue lines on the outer circles represent the sets of pure states that are compatible with a measurement resulting from $\mathcal{E}( \psi  ) $. The light blue area are the mixed states that are compatible with the same measurement. (a) The quantum states $ \mathcal{E}( \psi )$ and $\mathcal{E}( \psi^{\bot} )$ are so close that the measurement cannot distinguish between them. Therefore, $\left| \psi \right\rangle $ and $| \psi^{\bot} \rangle $ are fragile. (b) Since $\mathcal{E}( \psi )$ is well distinguishable from $\mathcal{E}( \psi^{\bot} )$, the respective pure states are called certifiable.}
\end{figure}

\subsection{Motivation and definition of certifiability}\label{sec:motiv_defn}

We consider a quantum state $|\psi\rangle$ and states orthogonal to $|\psi\rangle$. We assume that $|\psi\rangle$ is subjected to some noise process before a measurement can be performed.  For instance, this can arise due to imperfect state preparation, uncontrolled interactions with the environment or particle loss. The physical process is represented by a completely positive (cp), trace-preserving map ${\cal E}: \mathcal{D}(\mathcal{H})\rightarrow \mathcal{D}(\mathcal{H})$ \footnote{If $\mathcal{E}$ is applied to a pure state $| \phi \rangle $, we often use the shorthand notation ${\cal E}(\phi) \equiv {\cal E}(|\phi\rangle) \equiv {\cal E}(|\phi\rangle\!\langle\phi |)$.}. Due to finite precision of the measurement device and due to a bounded number of repetitions of the experiment, the accuracy of the measurement is bounded by some value $\delta >0$. We are concerned with the question whether such a measurement allows one to distinguish $\left| \psi \right\rangle $ from all orthogonal states in the presence of ${\cal E}$. If ${\cal E}(\psi)$ turns out to be indistinguishable from ${\cal E}(\phi)$ for some state $|\phi\rangle$ orthogonal to $|\psi\rangle$, one cannot determine whether the state $|\psi\rangle$ or the corresponding orthogonal state was initially present.  What is more (as we will show below) in such a situation one cannot even determine whether $|\psi\rangle$ was initially present or any incoherent mixture of $|\psi\rangle$ and $|\phi\rangle$. In this sense, the coherence of the state $|\psi\rangle$ is fragile.

We now define the quantitative measure which will be used to asses the fragility of quantum states under noise in the sense outlined above; see also figure \ref{fig:intuitivePicture} for a schematic illustration.

\begin{Def}[{\bf Certifiability under noise}]\label{def:Stability}
Consider  a noise channel ${\cal E}$. Let $\left| \psi \right\rangle \in \mathcal{H}$ be a normalised quantum state. We define the \emph{certifiability of $|\psi\rangle$ under ${\cal E}$} to be the minimal distance between $|\psi\rangle$ and any orthogonal state after the action of  ${\cal E}$:
\begin{equation}
\label{eq:1}
C(|\psi\rangle, {\cal E}) \mathrel{\mathop:}= \min_{\left| \phi \right\rangle \in  \mathcal{H}: \langle \psi | \phi\rangle  =0} D \left[ \mathcal{E}( \psi), \mathcal{E}( \phi  ) \right]
\end{equation}
where $D(\rho,\sigma) \mathrel{\mathop:}=  \lVert \rho-\sigma \rVert_{1}/2$ is the trace distance between density operators $\rho$ and $\sigma$. 
\end{Def}

We will often use the shorthand notation $C(|\psi\rangle, {\cal E})\equiv C(\psi)$ if the choice of noise channel is clear from the context.

The significance of the certifiability  is best understood by comparing the value $C(\psi)$ to the \emph{measurement accuracy} $\delta$ which is available in a given experiment. Recall that the trace distance is a measure for the distinguishability of density operators, i.e., it quantifies the optimal way of distinguishing two states via measurements. In particular, for all observables $A$  with spectral radius at most 1 it holds that \cite[Thm 9.1]{nielsen_quantum_2010}
\begin{equation}
\lVert \rho-\sigma \rVert_1 \geq \left| \langle A \rangle_{\rho} - \langle A \rangle_{\sigma} \right|.\label{eq:12}
\end{equation}
The quantity $C(\psi)$ is therefore a measure for the distinguishability between initially orthogonal states under the influence of noise described by ${\cal E}$. In particular, in any set-up where measurements are subject to an error $\delta> C(\psi)$  it is impossible to determine whether $|\psi\rangle$ was initially prepared or some state orthogonal to it. In this sense, the preparation of the state $|\psi\rangle$ cannot be certified.

Let $|\phi\rangle$ be a state orthogonal to $|\psi\rangle$ satisfying  $ D \left[ \mathcal{E}( \psi), \mathcal{E}( \phi  ) \right] = C(\psi)$. Consider the following one-parameter family of density operators:
\begin{equation}
\label{eq:9}
\rho(a):=
a \left| \psi \rangle\!\langle \psi\right| + (1-a) \left| \phi \rangle\!\langle \phi\right| ,
\end{equation}
with $a \in \left[ 0,1 \right]$. Then \begin{equation} D \left[ \mathcal{E}( \psi), \mathcal{E}( \rho(a)  ) \right]\leq  C(\psi)\end{equation} by using the triangle inequality. This implies that, in a scenario where measurements are subject to an error $\delta>C(\psi)$, one cannot determine if the pure state $|\psi\rangle$ \emph{or any density operator in the family $\rho(a)$} was present before the noise ${\cal E}$---note that, e.g., the state $\rho(1/2)$ has entropy $\log 2$ (compare to $|\psi\rangle$ which has \emph{zero} entropy), being an equal mixture of $|\psi\rangle$ and $|\phi\rangle$. In this sense, the preparation of a coherent quantum state cannot be certified. Notice also that the states $\rho(a)$ could be much less entangled than $\left| \psi \right\rangle $ itself, or they could even be fully separable (cf. the example of the GHZ state in section \ref{sec:basic_examples}). In such a case, also the presence of entanglement cannot be certified.

The certifiability $C(\psi)$ thus provides a measure to assess how fragile  pure states are under the influence of noise. In this work we will mostly be interested in the  behaviour of $C(\psi)$ for large systems in the presence of very small amounts of (local white) noise, giving rise to the notion of \emph{asymptotic (in)certifiability} in section \ref{sec:asymptotic_certifiability}. Before describing the latter notion, we first discuss the white noise model considered in this work.

\subsection{Noise model}\label{sec:noise_model}
Throughout the paper, we will focus on noise acting on individual particles, where we consider uncorrelated depolarisation channels. Given $\rho \in \mathcal{D}(\mathcal{H})$ and the depolarisation channel $\mathcal{E}^{(i)}:\mathcal{D}(\mathcal{H})\rightarrow\mathcal{D}(\mathcal{H})$ that acts on the $i$th particle, the state $\mathcal{E}^{(i)}(\rho)$  is defined as
\begin{equation}
\label{eq:8}
\mathcal{E}^{(i)}(\rho) = p \rho + (1-p) \mathrm{Tr}_i \rho \otimes \frac{\mathbbm{1}^{(i)}}{2},
\end{equation}
where $p \in \left[ 0,1 \right]$ is the noise parameter. The term $\mathrm{Tr}_i \rho \otimes \mathbbm{1}^{(i)}/2$ means that the $i$th particle is traced out and replaced by the complete mixture $\mathbbm{1}/2$. The total action of the depolarisation channel is the product of the individual channels acting on each qubit and is denoted by
\begin{equation}
\mathcal{E} \equiv \mathcal{E}^{(1)} \otimes \dots \otimes \mathcal{E}^{(N)}.
\end{equation}

The uncorrelated depolarisation channel---often also called white noise channel---is a natural choice for the study of stability of many-body states. Its action is invariant under single-body unitary rotations and, since it is uncorrelated, decoherence free subspaces for entangled states do not exist. When $p= 0$, the channel maps every quantum state to the complete mixture $\mathbbm{1}^{\otimes N}/2^N$. Physically, the depolarisation channel can arise in the context of thermal baths with infinite temperature to which the qubits individually couple. Then, one assigns $p= e^{-\gamma t}$, where $\gamma$ is the depolarisation rate and $t$ the evolving time. Alternatively, the depolarisation channel can be seen as a loss channel, where for measurements every qubit accounts only with probability $p$.

Henceforth we consider the situation where a fixed, very small amount of noise acts on each particle. This corresponds to a white noise parameter $p$ for some   constant $p$  (i.e.~independent of the number of particles $N$) which is arbitrarily close to 1.

\subsection{Large-scale systems: asymptotic certifiability}\label{sec:asymptotic_certifiability}

As discussed in section \ref{sec:motiv_defn}, the main importance of the measure $C(\psi)$ lies in its comparison with the accuracy with which expectation values of measurements can be determined in an experiment.  The achievable accuracy of a measurement is fundamentally limited by measurement statistics. Given the fact that quantum mechanics is a probabilistic theory, expectation values of any observable can only be determined with an error $1/\sqrt m$, where $m$ are the number of repetitions of the experiment. If we are interested in large system sizes $N$, for all practical purposes the smallest achievable error of the measurements will behave as $1/$poly$(N)$ for some polynomial poly$(N)$, due to the fact that at most a polynomial number of repetitions of the experiment are possible in principle.  Imperfect measurement apparatuses or finite resolution of the measurement device may even lead to a constant achievable error $\delta$. \footnote{In some scenarios, it is possible to treat the noise induced by the measurement apparatus mathematically as noise acting on the state before the measurement. For example, a measurement with $\sigma_z$ that only probabilistically gives the correct answer can be modelled as a bit-flip error on the state of interest. Then, $\delta$ purely depends on the measurement statistics, while $\mathcal{E}$ covers the imperfections of state preparation, decoherence and apparatus noise.}

In the following, we will be interested in the scaling of the certifiability $C(\psi)$ with the system size $N$, as follows:

\begin{Def}\label{def:asymp_cert}[{\bf Asymptotically certifiable and incertifiable quantum states}] Let ${\cal E}$ denote local white noise with parameter $p$ as described in section \ref{sec:noise_model}. A state $|\psi\rangle$ is called \emph{asymptotically incertifiable under local white noise} if the quantity $C(|\psi\rangle, {\cal E})$ tends to zero faster than any inverse polynomial in $N$, i.e., if \be C(|\psi\rangle, {\cal E})=O(p^{N^\alpha})\ee for some $\alpha >0$. Conversely, a state with \be C(|\psi\rangle, {\cal E}) \geq  1/\mbox{poly}(N)\ee for some polynomial will be called asymptotically certifiable under local white noise. Henceforth, for the sake of brevity we will simply refer to ``certifiable'' and ``incertifiable'' states.
\end{Def}
If a state is found to be incertifiable in the sense of definition \ref{def:asymp_cert}, then for large system sizes it is essentially impossible to discriminate between this state and some orthogonal state, or in fact any incoherent mixture of these two states [recall  (\ref{eq:9})] with any measurement that is repeated at most poly$(N)$ times.

As we will see below, the asymptotic behaviour of the certifiability may be very different, also qualitatively, for various state families. In other words, this measure provides a criterion  of discriminating between ``fragile'' and ``stable'' large-scale coherent quantum states under the influence of noise. We further remark that several measures and quantities of interest which are often considered to determine the quality of a prepared state do not allow for a meaningful distinction between fragile and non-fragile large-scale states. Consider for instance the purity of the noisy quantum state, $\mathrm{tr}[{\cal E}(\psi)^2]$, or the fidelity of ${\cal E}(\psi)$ with the initial desired pure state $|\psi\rangle$. Both quantities can be shown to decay exponentially fast with the number of particles $N$ for \emph{all} states if each particle is affected by, say,  local white noise. The decay rate may differ, but in the limit of large particle number $N$ that we are interested in, these quantities will be exponentially small whenever noise is present. This implies on the one hand that they show the same qualitative behaviour for all states, and hence to not allow to clearly distinguish fragile and non-fragile states. Furthermore, the resolution of a realistic measurement device will not be sufficient to determine these quantities in a meaningful way as this would require exponential precision.

Regarding definition \ref{def:asymp_cert}, we finally remark that limited accuracy of measurement devices or bounded number of measurements may render certain states that are certifiable in above sense as incertifiable with respect to measurements with a limited precision. In turn, some states that are incertifiable in above sense may in fact be certifiable as long as only a system of restricted size $N$ is considered, even if a measurement device with a given finite precision is considered. In the following, we will however not be interested in such cases, but rather provide general results that are concerned with principal bounds and the scaling behaviour with system size.

\section{Basic examples of certifiable and incertifiable states}\label{sec:basic_examples}
To make the reader more familiar with the notion of (in)certifiability, we consider two basic examples.                                                                                                                                                                               \begin{Exp}[{\bf Product States}]
\label{Ex:productstate}
Here, we explicitly show that the certifiability of any $N$-qubit product state scales as $1/$poly$(N)$ i.e.~all product states are asymptotically certifiable. In particular, the trace distance between a noisy product state and any other noisy orthogonal state is at least of the order of $p/N$. Furthermore the observable $A$ which allows to distinguish a noisy product state from any noisy orthogonal state can always be chosen to be 1-local i.e. $A=\sum_i A_i$ where each observable $A_i$ acts on one single qubit.                                                                                                                                                                                                    To show this, we consider the state $\left| N,0 \right\rangle = \left| 0 \right\rangle^{\otimes N} $ where $\left| 0 \right\rangle $ is the eigenstate of the Pauli operator $\sigma_z$ with eigenvalue $+1$. The eigenstate with eigenvalue $-1$ is denoted by $\left| 1 \right\rangle.$  The state $\left| N,0 \right\rangle $ is therefore an eigenstate of \footnote{The expression $\sigma^{(i)}_z$ means that $\sigma_z$ acts on the $i$th particle only (i.e., $\sigma_z^{(i)} \equiv \mathbbm{1}^{\otimes i-1}\otimes\sigma_z\otimes \mathbbm{1}^{\otimes N-i-1}$).} \be J_z = \frac{1}{2} \sum_{i=1}^N \sigma^{(i)}_z\ee  and one can easily see that it is the unique eigenstate with maximum eigenvalue $N/2$. The eigenvalues of $J_z$ are $\lambda_{k} = N/2-k$ where  $k = 0,\dots,N$. The eigenspaces of $J_{z}$ with eigenvalues $\lambda_k$ are generally degenerate and have dimensions $\binom{N}{k}$. Hence, the second highest eigenvalue of $J_z$ equals $N/2-1$; an instance of a respective eigenstate is the W state \cite{dur_three_2000} \be | N,1 \rangle &=& \frac{1}{\sqrt{N}} \sum_{i=1}^N | 1 \rangle\!\langle 0|^{(i)}| N,0 \rangle. \ee We now use the fact that one can distinguish $\left| N,0 \right\rangle $ from all other orthogonal states by measuring $J_z$. Let $\left| \phi \right\rangle $ be the state that minimises $C(\left| N,0 \right\rangle )$ of  (\ref{eq:1}). Then, using  (\ref{eq:12}), one estimates
                                                                                                  \begin{equation}
\label{eq:11}
\begin{split}
  C( | N,0  \rangle ) &= \frac{1}{2}\lVert \mathcal{E} (| N,0 \rangle\!\langle N,0 | -  | \phi \rangle\!\langle \phi |  ) \rVert_1 \\ & \geq \frac{1}{N}                               | \mathrm{Tr} [ J_{z} \mathcal{E} (  | N,0 \rangle\!\langle N,0 |  -  | \phi \rangle\!\langle \phi |  ) ] |\\ & = \frac{1}{N}| \mathrm{Tr} [ \mathcal{E} (J_{z}   )  (   | N,0 \rangle\!\langle N,0 |-  | \phi \rangle\!\langle \phi | )  ] |.
\end{split}
\end{equation}
The last equation holds because for the present noise model, the expectation value of a perfect measurement $J_z$ under noisy quantum states is equal to the expectation value of a noisy measurement under perfect states. This is true because the Kraus operators (which are the Pauli operators) are hermitian. Since $\mathcal{E}(\mathbbm{1}) = \mathbbm{1}$ and $\mathcal{E}(\sigma) = p \sigma$ for any Pauli operator, one easily finds that $\mathcal{E}(J_z) = p J_z$. Therefore, one has                                                              \begin{equation}                                                                       \label{eq:13}
C(\left| N,0 \right\rangle ) \geq \frac{p}{N} \left| \langle J_z \rangle_{\left| N,0 \right\rangle } - \langle J_z \rangle_{\phi}\right|.
\end{equation}
As already noted, the expectation value of $J_z$ for $\left| N,0 \right\rangle $ equals $N/2$, while for any orthogonal state the expectation value is at most $N/2-1$. Hence,  (\ref{eq:13}) can be further estimated and we find
                                                                                                  \begin{equation}
\label{eq:14}
C(\left| N,0 \right\rangle ) \geq \frac{p}{N}.
\end{equation}

This means that the product state $\left| N,0 \right\rangle $ is asymptotically certifiable with respect to definition 2 and therefore, even in the presence of white noise, we can always determine with poly-many measurements whether the quantum state $\left| N,0 \right\rangle $ was prepared or any orthogonal state [cf.~figure \ref{fig:intuitivePicture}(b)]. An instance of a measurement that can be used for this task is $J_z$. 
Note that  (\ref{eq:1}) is invariant under local unitary rotation. This example therefore shows that all product states are stable. A more general result is proven in section \ref{sec:unique-ground-states}.
\end{Exp}
\begin{Exp}[{\bf GHZ states}]
 \label{Ex:GHZ}
  The second example is the multipartite GHZ state
\begin{equation}
\label{eq:15}
| \mathrm{GHZ}  \rangle = \frac{1}{\sqrt{2}}  (  | 0  \rangle ^{\otimes N} +  | 1  \rangle^{\otimes N}   ).
\end{equation}
This quantum state is often referred to as a Schr\"odinger cat state (see section \ref{sec:macr-superp-are}). In the following, we show that this state is asymptotically incertifiable. We therefore consider an ``orthogonal GHZ state'' \be | \mathrm{GHZ}^{\bot}\rangle = \frac{1}{\sqrt{2}}( \left| 0 \right\rangle ^{\otimes N} - \left| 1 \right\rangle^{\otimes N} ).\ee 
The certifiability $C(\mathrm{GHZ})$ of the GHZ state  can be bounded from above by the trace distance between $| \mathrm{GHZ} \rangle $ and $| \mathrm{GHZ}^{\bot} \rangle $ in the presence of white noise. A straightforward calculation results in
                                                                                                                                                                                                                                                                                             \begin{equation}
\label{eq:10}
 \begin{split}
C(\mathrm{GHZ}) & \leq \frac{1}{2}\lVert \mathcal{E}(
|  \mathrm{GHZ} \rangle)  -\mathcal{E}(
|  \mathrm{GHZ}^{\bot}\rangle) \rVert_{1}\\ & =
\frac{1}{2} \lVert \mathcal{E} (  | 0  \rangle\! \langle 1 | ^{\otimes N}   )+  \mathcal{E} (  | 1 \rangle\!\langle 0 | ^{\otimes N}   ) \rVert_{1}\\ & \leq
\lVert \mathcal{E} (  | 0 \rangle\!\langle 1 |   ) \rVert^N_1 = p^N.
\end{split}
\end{equation}
For the last estimate one uses the triangle inequality and the fact that $\lVert  \mathcal{E}\left( \left| 1 \rangle\!\langle 0\right| \right) \rVert_1 = \lVert \mathcal{E}\left( \left| 0 \rangle\!\langle 1\right|  \right) \rVert_1$. One sees that the GHZ state is asymptotically incertifiable because the trace distance to its orthogonal counterpart $| \mathrm{GHZ}^{\bot}\rangle $ decays with at least $p^N$, that is, exponentially with system size (cf.~figure \ref{fig:intuitivePicture}a).  Note that the above derivation holds for every noise parameter $p$ i.e. \emph{$p$ may be arbitrary close to 1}. This means that it is not possible for whatever tiny amount of noise to experimentally distinguish between the GHZ state $| \mathrm{GHZ}^{\bot} \rangle$. Recalling the discussion of  (\ref{eq:9}), this also applies that, in the presence of small amounts of noise, the GHZ state cannot be distinguished from the mixture \be \frac{1}{2}| \mathrm{GHZ} \rangle\!\langle \mathrm{GHZ}| + \frac{1}{2}| \mathrm{GHZ}^{\bot} \rangle\!\langle \mathrm{GHZ}^{\bot}|  = \frac{1}{2}\left| 0 \rangle\!\langle 0\right|^{\otimes N} + \frac{1}{2} \left| 1 \rangle\!\langle 1\right|^{\otimes N} ,\ee which is a separable state! Hence, not only the coherences of the GHZ state are incertifiable, but also other quantum properties like entanglement.
                                                                                                  \end{Exp}

These two examples show that definition 2 gives rise to a criterion that makes qualitatively very different statements about the stability of distinct families many-qubit quantum states. In sections \ref{sec:unique-ground-states} and \ref{sec:macr-superp-are}, we extend these examples to large classes of quantum states.

\section{Unique ground states of local, gapped Hamiltonians are certifiable}
                                                                                                                                                                                                                                                                                              \label{sec:unique-ground-states}
In this section, we discuss an important class of quantum states that are asymptotically certifiable. It is the set of all unique ground states of $k$-local, gapped Hamilton operators. By $k$-local it is meant that if we write the Hamiltonian in the Pauli basis, the maximal weight of every term is $k$, where $k=O(1)$. In other words, we consider only interactions between a constant number of qubits (without restricting ourselves to a specific spatial geometry). Here, the term gapped means that the energy difference $\Delta$ between the ground state and the first excited state is at least $O(1/ \mathrm{poly}(N))$. Finally, a ground state is unique if the ground state space is nondegenerate.
\begin{Pro}
\label{pro:unique-ground-states}
Let $\left| \psi \right\rangle \in  \mathcal{H}$ be the unique ground state of a $k$-local, gapped Hamiltonian $H$. Then $|\psi\rangle$ is asymptotically certifiable under local white noise. That is, $C(\psi) \geq 1/ \mathrm{poly}(N)$.
                                                                                                  \end{Pro}
\begin{proof}
 Let $\sigma_0$ denote the single-qubit identity matrix and let $\sigma_1, \sigma_2$ and $\sigma_3$ denote the Pauli matrices. First, we write $H$ in the Pauli basis
\begin{equation}
\label{eq:2}
H = \sum_{i_1,\dots,i_N =0}^3 \alpha_{i_1,\dots,i_N}\sigma_{i_1}\otimes \dots \otimes \sigma_{i_N},
\end{equation}
with $\alpha_{i_1,\dots,i_N} \in \mathbbm{R}$. From now on, we abbreviate $\mathbf{i} \equiv i_1,\dots, i_N$ and $\sigma_{\mathbf{i}} \equiv \sigma_{i_1}\otimes \dots \otimes \sigma_{i_N}$. The weight $w_{\mathbf{i}}$ of $\sigma_{\mathbf{i}}$ is defined the be the number of qubits on which this operator acts nontrivially. Let us first consider the case where $w_{\mathbf{i}} = k \in \mathbbm{N}$ for all $\mathbf{i}$ i.e.~each $\sigma_{\mathbf{i}}$ acts on precisely $k$ qubits. Later we will extend the proof to the general case $w_{\mathbf{i}} \leq k$.
Let $\left| \phi \right\rangle \in \mathcal{H}$ be an arbitrary  quantum state satisfying $\langle \psi | \phi\rangle =0$. Then one has
\begin{equation}
\label{eq:3}
\begin{split}
D [ \mathcal{E}(\psi), \mathcal{E}(\phi)] \geq \frac{1}{2} \left| \mathrm{Tr} \left[ H\mathcal{E} \left( \left| \psi \rangle\!\langle \psi\right| - \left| \phi \rangle\!\langle \phi\right|   \right)\right] \right|
\end{split}
\end{equation}
For this step we have to assume that the spectral radius of $H$ is at most unity. If not, we have to rescale it. Physically, it is reasonable to assume interaction strengths between the particles that are independent of the system size; therefore, $\alpha_{\mathbf{i}} = O(1)$. For local $H$, the number of $\sigma_{\mathbf{i}}$ is polynomial in $N$ and the spectral radius of any $\sigma_{\mathbf{i}}$ is unity. Therefore, the rescaling factor is at most $1/poly(N)$ which does not change the result of the proposition. So without loss of generality, the spectral radius of $H$ is set to unity.
Next, the expectation value of $H$ for a noisy state is identical to the expectation value of the noisy $H$ for a noiseless state, that is, for any $\rho$, $\mathrm{Tr} \left[ H\mathcal{E} \left( \rho   \right)\right] = \mathrm{Tr} \left[ \mathcal{E}(H)\rho\right]$. For a single qubit, the action of $\mathcal{E}$ is easy to calculate in the Pauli basis. One has that $\mathcal{E}(\sigma_0) = \sigma_0$ and $\mathcal{E}(\sigma_{i}) = p\sigma_{i}$ for every $i=1, 2, 3$; hence $\mathcal{E}(\sigma_{\mathbf{i}}) = p^{w_{\mathbf{i}}} \sigma_{\mathbf{i}}$. Since $w_{\mathbf{i}} = k$ for all $\mathbf{i}$, one has $\mathcal{E}(H) = p^k H$. Hence, one finds that
                                                                                                   \begin{equation}
\label{eq:4}
|\mathrm{Tr}  [ \mathcal{E}(H)  (  | \psi \rangle\!\langle \psi | -  | \phi \rangle\!\langle \phi |   ) ]| = p^k  |\langle H \rangle_{\psi} - \langle H \rangle_{\phi}  |.
\end{equation}
Since $\left| \phi \right\rangle $ is orthogonal to $|\psi\rangle$, the difference between the two expectation values in  (\ref{eq:4}) is at least the energy gap $\Delta$. Consequently, we find \be C(\psi)\geq p^k \Delta/2.\ee
This completes the proof of the special case where $w_{\mathbf{i}} = k$ for all $\mathbf{i}$.
In general, one has $w_{\mathbf{i}} \leq k$. Then, we enlarge the Hilbert space to $\mathcal{\tilde{H}}$ by adding $k$ additional qubits i.e. $\tilde{\cal H}$ is a system of $N+k$ qubits. We define
\begin{equation}
\label{eq:5}
\begin{split}
| \tilde\psi  \rangle &=  | \psi  \rangle
\otimes  | 0  \rangle ^{\otimes k} \in \mathcal{\tilde{H}}\\
|\tilde\phi  \rangle &=  | \phi  \rangle
\otimes  | 0  \rangle ^{\otimes k}\in \mathcal{\tilde{H}}\\
\tilde{H} &= \sum_{\mathbf{i}} \alpha_{\mathbf{i}} \, \sigma_{\mathbf{i}} \otimes \sigma_z^{\otimes (k-w_{\mathbf{i}})}\otimes \mathbbm{1}^{\otimes w_{\mathbf{i}}}.
\end{split}
\end{equation}
In the above definition of $\tilde H$, the idea is to append $\sigma_z$ operators such that now the weight of every term of $\tilde{H}$ is precisely $k$. Then, we can use similar arguments as before. One has
\be  D[ \mathcal{E}(\psi), \mathcal{E}(\phi) ] &=& D[ \mathcal{E}(\tilde\psi), \mathcal{E}(\tilde\phi)]\geq  \frac{1}{2}  | \mathrm{Tr}  [ \tilde{H} \mathcal{E}  ( | \tilde\psi  \rangle \!   \langle \tilde\psi | -  | \tilde\phi  \rangle \!  \langle \tilde\phi |   ) ]  | \\  &=& \frac{p^k}{2} |   \langle \tilde{H} \rangle_{\tilde{\psi}} - \langle \tilde{H} \rangle_{\tilde{\phi}}| \geq \frac{1}{2}p^k \Delta \ee
Here the first identity is obtained by using that $D(\rho, \rho')= D(\rho\otimes \sigma, \rho'\otimes\sigma)$ for all states $\rho$, $\rho'$ and $\sigma$; the last estimate is valid because $\langle \tilde{H} \rangle_{\widetilde{\psi}} = \langle H \rangle_{\psi}$ and respectively for $| \phi\rangle $.
 \end{proof}

Note that the set of unique ground states of local gapped Hamiltonians includes all product states and even tensor products of general quantum states on groups of qubits where the group size is independent of the system size. A large class of unique ground states of local gapped Hamiltonians are local graph states \cite{hein_entanglement_2005}, i.e.~graph states corresponding to a graph with bounded degree. More precisely, graph states are quantum states that are associated with an undirected graph $G$. Every vertex of the graph represents a qubit and the edges between the vertices define the neighbourhood. The graph state $|G\rangle$ is then defined as the unique state that is eigenstate to the so-called generating stabilisers \be K_a = \sigma_{x}^{(a)} \prod_{b\in N_a} \sigma_z^{(b)} \quad (a = 1,\dots,N), \ee where $N_a$ is the neighbourhood of vertex $a$. A graph is called local if every vertex has only a constant number of neighbours. The corresponding gapped Hamiltonian which has $|G\rangle$ as unique ground state is $-\sum_{a = 1}^N K_a$. Prominent examples in this category are the cluster states \cite{briegel_persistent_2001} where every qubit has two neighbours in an one-dimensional lattice or four neighbours in a two-dimensional lattice, respectively.

Sometimes, an educated guess for a suitable Hamiltonian suffices to show stability for a given quantum state $|\psi\rangle$. A general approach to construct a Hamiltonian is to identify a commuting set of local, gapped operators $H_1,\dots,H_l$ with the property that  $\left| \psi \right\rangle$ is the only state that is a ground state of all operators of this set simultaneously. Then, the operator $H:=H_1 + \dots + H_l$ is  local and gapped and has $|\psi\rangle$ as a unique ground state. This method is standardly used to construct Hamiltonians for graph states, as discussed above. Next we apply this approach to the Dicke states:

\begin{Exp}[{\bf Dicke states}]
\label{exp:DickeStates}
Here we show that all Dicke states are unique ground states of local, gapped Hamiltonians and hence asymptotically certifiable owing to proposition \ref{pro:unique-ground-states}. In example \ref{Ex:productstate}, we encountered already certain Dicke states, namely the product state $\left| N,0 \right\rangle $ and the W state $\left| N,1 \right\rangle $ as the equally weighted superposition of all permutations of $\left| 1 \right\rangle \otimes \left| 0 \right\rangle ^{\otimes (N-1)}$. To define general Dicke states, consider again the $z$ component of the total angular momentum $J_z$. The Dicke state $\left| N,k \right\rangle $, where $k\in\{0, \dots, N\}$,  is defined as the unique eigenstates of $J_z$ with eigenvalue $\lambda_{k} = N/2 - k$ which is invariant under any particle permutation (i.e., $| N,k \rangle $ is the symmetrised vector of $\left| 0 \right\rangle ^{\otimes N-k}\otimes \left| 1 \right\rangle ^{\otimes k}$).
Consider the square of the total angular momentum $J^2 = J_x^2 + J_y^2 + J_z^2$. All Dicke states are eigenstates of $J^2$ with the maximal eigenvalue $\mu_{\mathrm{max}} = N/2(N/2+1)$ and (invoking Schur-Weyl-duality \cite{goodman_representations_1998}) linear superpositions of Dicke states are the only states with eigenvalue $\mu_{\mathrm{max}}$. Now fix an arbitrary $k$ and define \be \tilde{J}_z =  J_z -\left( N/2 - k \right) \mathbbm{1}^{\otimes N}.\ee Then $\tilde{J}_z^2$ is a positive semi-definite operator. It is easily verified that the ground space of $\tilde{J}_z^2$ is spanned by the set of computational basis states which have precisely $k$ slots equal to $|1\rangle$. As a result, the Dicke state $\left| N,k \right\rangle $ is the unique state which is a ground state of  $-J^2$ and  $-\tilde{J}_z^2$ simultaneously.  Hence, $\left| N,k \right\rangle $  is the unique ground state of $ -J^2 -\tilde{J}_z^2$. In the theory of angular momentum algebra, one finds that $J^2$ and $J_z$ commute. Therefore, also $J^2$ and $\tilde{J}_z^2$ commute. Finally, the spectra of $J^2$ and $\tilde{J}_z^2$ are gapped, which again follows from the commutation algebra. Hence, $ -J^2 -\tilde{J}_z^2$ is also gapped. We have thus showed that every Dicke state is the unique ground state of a local, gapped Hamiltonian.

\end{Exp}
\section{Macroscopic superpositions are incertifiable}
\label{sec:macr-superp-are}

Here, we show that all ``macroscopic superpositions'' are asymptotically incertifiable: for every macroscopic superposition there always exists an orthogonal quantum state (which  also turns out to be a macroscopic superposition) such that the action of very small amounts of white noise makes them indistinguishable by efficient means. 

Macroscopic superpositions in the spirit of Schr\"odinger's cat are superpositions of two ``classical states'' $|\psi_0\rangle$ and $|\psi_1\rangle$ that give rise to a ``macroscopic quantum state'' $|\psi_0\rangle + |\psi_1\rangle$. Criteria for macroscopicity have been put forward in   \cite{korsbakken_measurement-based_2007,marquardt_measuring_2008,bjork_size_2004} and have been compared and unified in   \cite{frowis_measures_2012}. We refer the interested reader to \cite{frowis_measures_2012} for a detailed discussion of measures for macroscopicity and macroscopic superpositions.
What is important in our context is that the results of  \cite{frowis_measures_2012} imply that all superpositions that are macroscopic due to the definition put forward in  \cite{bjork_size_2004,marquardt_measuring_2008} are also macroscopic according to  \cite{korsbakken_measurement-based_2007}. Here we show that \emph{all macroscopic superpositions in the general sense of \cite{korsbakken_measurement-based_2007} are asymptotically incertifiable}.

This result connects to decoherence theory (see \cite{zurek_pointer_1981,zurek_preferred_1993,zurek_decoherence_2003,joos_decoherence_2003,schlosshauer_decoherence:_2007} and references within \cite{zurek_decoherence_2003}), where the (exponential) decay of quantum coherences of e.g. Schrödinger cat states is a well-known phenomenon which is put forward as a potential explanation for the classical appearance of the macroscopic world. This is the basic idea of the environmentally induced selection of classical states. Whereas prior work focused on examples, here we investigate this relation mathematically more rigorously and systematically, in particular adopting a general definition of macroscopic superpositions. As we will show in proposition \ref{pro:macr-superp-unstable} below, we find that there is an inherent connection between macroscopic superpositions and fragility under noise i.e.~the former implies the latter - even though our notion of macroscopic superpositions a prior makes no reference to decoherence processes. 

We start by briefly introducing the relevant definitions of  \cite{korsbakken_measurement-based_2007}. Consider two orthogonal quantum states $\left| \psi_0 \right\rangle, \left| \psi_1 \right\rangle \in \mathcal{H}$. Since they are orthogonal, one can certainly distinguish them via a suitable measurement (in the perfect setting of zero noise, in which we work for the time being). The goal is now to split the qubits into a maximal number of groups such that if one performs measurements on only one of the groups, it is still possible to distinguish $\left| \psi_0 \right\rangle $ from $\left| \psi_1 \right\rangle $ with high probability $1-\epsilon$ (for the value of $\epsilon$, see the discussion in \ref{sec:fagil-macr-superp}). This maximal number of groups is then called ``effective size'' $N_{\mathrm{eff}}$. The most conservative definition is then to call a superposition $\left| \psi \right\rangle \propto \left| \psi_0 \right\rangle + \left| \psi_1 \right\rangle $ macroscopic if $N_{\mathrm{eff}} = O(N)$, that is, the average group size is independent of the system size $N$. A famous example of a macroscopic superposition is the GHZ state with $\left| \psi_0 \right\rangle  = \left| 0 \right\rangle ^{\otimes N}$ and $\left| \psi_1  \right\rangle = \left| 1 \right\rangle ^{\otimes N}$. There, it suffices to measure one particle (with $\sigma_z$) to distinguish between $\left| \psi_0 \right\rangle $ and $\left| \psi_1 \right\rangle $ perfectly; therefore, $N_{\mathrm{eff}} = N$.

Next we prove that all macroscopic superpositions in the above sense are asymptotically incertifiable. First we deal with the case of perfect distinction (i.e., $\epsilon = 0$).

 \begin{Pro}
\label{pro:macr-superp-unstable}
Let $\left| \psi \right\rangle  = 1/\sqrt{2} \left( \left| \psi_0 \right\rangle + \left| \psi_1 \right\rangle  \right) \in \mathcal{H}$ be a normalised quantum state with $\langle \psi_0 | \psi_1\rangle  = 0$. Assume that one can partition the $N$ qubits into $N_{\mathrm{eff}} \in \left\{ 1,\dots,N \right\}$ groups such that measuring one group allows one to distinguish perfectly between $\left| \psi_0 \right\rangle $ and $\left| \psi_1 \right\rangle $. Then, one has
\begin{equation}
\label{eq:16}
C(\psi) \leq q^{N_{\mathrm{eff}}}
\end{equation}
with \be q = 1-(1-p)^{N/N_{\mathrm{eff}}}.\ee If $N_{\mathrm{eff}} = O(N)$---that is, if $\left| \psi \right\rangle $ is a macroscopic superposition---then $\left| \psi \right\rangle $ is asymptotically incertifiable.
\end{Pro}
\begin{proof}
It suffices to show that the trace distance between the noisy $\left| \psi \right\rangle $ and the noisy orthogonal state \be | \psi^{\bot} \rangle  := \frac{1}{\sqrt{2} } \left( \left| \psi_0 \right\rangle - \left| \psi_1 \right\rangle  \right)\ee is upper bounded by $q^{N_{\mathrm{eff}}}$. One has
\begin{equation}
\label{eq:17}
\begin{split}
 C(\psi) &\leq \frac{1}{2} \lVert \mathcal{E} (  | \psi
 \rangle\!\langle \psi | -  | \psi^{\bot}
 \rangle\!\langle \psi^{\bot} |  ) \rVert_1 \\ &
    =  \frac{1}{2} \lVert \mathcal{E} (  | \psi_0 \rangle\!\langle \psi_1 | +  | \psi_1 \rangle\!\langle \psi_0 |   ) \rVert_{1} \\ & \leq \lVert \mathcal{E} (  | \psi_0 \rangle\!\langle \psi_1 |    ) \rVert_{1},
 \end{split}
 \end{equation}
  which means that the decay of the off-diagonal element gives an upper bound on the certifiability of $\left| \psi \right\rangle $.
 To avoid an overloaded notation, we assume that the size of every group equals $m = N/N_{\mathrm{eff}}$. The generalisation to different group sizes is straightforward. We now introduce for every group a ``group-local'' depolarisation noise $\mathcal{E}_g$ which acts on a group of $m$ qubits as \be \mathcal{E}_g(\rho) = q \rho + (1-q)\  \mathrm{Tr}\rho \ \left( \frac{\mathbbm{1}}{2} \right)^{\otimes m}.\ee Observe that we can rewrite the local depolarisation channel as \be \mathcal{E}  \equiv \mathcal{E}^{(1)} \otimes \dots \otimes \mathcal{E}^{(N)}= \left( \mathcal{\tilde{E}} \circ \mathcal{E}_g \right)^{\otimes N_{\mathrm{eff}}}\ee with the ``correction'' map $\mathcal{\tilde{E}}$ acting on a group of $m$ qubits as
\begin{equation}
\mathcal{\tilde{E}}(\rho) =
\frac{1}{1-(1-p)^m}\sum_{k=0}^{m-1}p^{m-k}(1-p)^k  \sum_{i_1 <
   \dots < i_k} \mathrm{Tr}_{i_1,\dots,i_k} \rho \otimes
\frac{\mathbbm{1}^{\otimes k}}{2^k}.\label{eq:18}
\end{equation}
In \ref{sec:lemma-proof-prop}, it is shown that $\mathcal{\tilde{E}}$ is a cp map. This property allows to invoke the contractivity of the trace norm, which implies that the trace norm of the off-diagonal element of  (\ref{eq:17}) does not decrease if we only consider $\mathcal{E}_g^{\otimes N_{\mathrm{eff}}}$ instead of $\mathcal{E}$. In mathematical terms, using Theorem 9.2 of  \cite{nielsen_quantum_2010} we find
\begin{equation}
\label{eq:19}
 \lVert \mathcal{E}\left( \left| \psi_0 \rangle\!\langle \psi_1\right|   \right) \rVert_{1} \leq \lVert \mathcal{E}_g^{\otimes N_{\mathrm{eff}}}\left( \left| \psi_0 \rangle\!\langle \psi_1\right|   \right) \rVert_{1}.
\end{equation}

 Let us now consider the action of $\mathcal{E}_g$ on one group (without loss of generality the first $m$ qubits) and in particular the effect of tracing out all qubits of this single group. Since $\left| \psi_0 \right\rangle $ and $\left| \psi_1 \right\rangle $ are perfectly distinguishable by measuring this group, there must exist an $m$-qubit observable $A$ such that $A\otimes \mathbbm{1}^{\otimes N-m} \left| \psi_i \right\rangle \equiv \tilde{A} \left| \psi_i \right\rangle = (-1)^i \left| \psi_i \right\rangle $ for $i=0,1$. In particular, we can construct $A$ such that the only eigenvalues of $A$  are $1$ and $-1$, which implies that $A^2 = I$. Then, one easily verifies that, for every $j = \left\{ 1,\dots,m \right\}$, one has  
 \begin{equation}
 \label{eq:136}
\begin{split}
\mathrm{Tr}_j (\left| \psi_0 \rangle\! \langle \psi_1\right|) & =
\frac{1}{2}\mathrm{Tr}_j (\left| \psi_0 \rangle\!\langle
\psi_1\right|) + \frac{1}{2}\mathrm{Tr}_j (\left| \psi_0
\rangle\!\langle \psi_1\right|)\\ &
  = \frac{1}{2}\mathrm{Tr}_j (\left| \psi_0 \rangle\!\langle
 \psi_1\right|) + \frac{1}{2}\mathrm{Tr}_j (\tilde{A}\left| \psi_0 \rangle\!\langle
  \psi_1\right|\tilde{A})\\ &
  = \frac{1}{2}\mathrm{Tr}_j (\left| \psi_0 \rangle\!\langle
  \psi_1\right|) - \frac{1}{2}\mathrm{Tr}_j (\left| \psi_0 \rangle\!\langle
  \psi_1\right|) = 0.
 \end{split}
\end{equation}
 This shows that tracing out any of the $N_{\mathrm{eff}}$ groups maps the off-diagonal element $\left| \psi_0 \rangle\!\langle \psi_1\right| $ to the zero-operator. The action of the group-local white noise is therefore given by \be \mathcal{E}_g \otimes I (\left| \psi_0 \rangle\!\langle \psi_1\right| ) = q \left| \psi_0 \rangle\!\langle \psi_1\right| \ee where $I$ is the identity operation on the rest of the space. Putting everything together, we find
 \begin{equation}
\label{eq:20}
C(\psi)\leq\lVert \mathcal{E}_g^{\otimes N_{\mathrm{eff}}}\left( \left| \psi_0 \rangle\!\langle \psi_1\right|   \right) \rVert_{1} = q^{N_{\mathrm{eff}}}.
\end{equation}
\end{proof}

   Note that we have thus shown that a large-scale coherent macroscopic superposition $ ( | \psi_0\rangle + | \psi_1\rangle)/\sqrt{2}$ is effectively indistinguishable from the incoherent mixture \be \frac{1}{2}| \psi_0\rangle\langle\psi_0| + \frac{1}{2}| \psi_1\rangle\langle\psi_1|\ee in the presence of small amounts of local white noise. 

 Also, together with proposition \ref{pro:unique-ground-states} we conclude that no macroscopic superposition can be the unique ground state of a local, gapped Hamiltonian.

   Apart from the GHZ state, there are other quantum states that are macroscopic superpositions in the sense of  \cite{korsbakken_measurement-based_2007}; all of which are also incertifiable. We summarise two instances.
\begin{Exp}
                                                                                                  \label{ex:macr-superp}
Consider a local graph state $\left| G \right\rangle $ \cite{hein_entanglement_2005}, that is, a graph state where the neighbourhood of every vertex is independent of the system size. The superposition $\left| G \right\rangle + \sigma_z^{\otimes N} \left| G \right\rangle $ is macroscopic because one can locally distinguish between $\left| G \right\rangle $ and $\sigma_z^{\otimes N} \left| G \right\rangle $ by measuring the (local) stabiliser operators. The effective size is roughly the system size divided by average weight of the stabilisers. Note that with proposition \ref{pro:unique-ground-states}, this macroscopic superposition is necessarily represented by a nonlocal graph. A concrete example for $\left| G \right\rangle $ are Cluster states \cite{briegel_persistent_2001} (with $N_{\mathrm{eff}} = N/3$ for the one-dimensional Cluster state). Note also that, if $G$, is the empty graph, then $\left| G \right\rangle + \sigma_z^{\otimes N} \left| G \right\rangle $ is the GHZ state (up to a local basis change).

Another instance is the superposition of the W state $\left| N,1 \right\rangle $ with the ``inverted'' W state $\left| N,N-1 \right\rangle $. Measuring one particle gives a success probability of $1-1/N$ for the distinction of those two quantum states. Therefore, the effective size equals $N$ in the limit of large $N$.

\end{Exp}

The situation is more complicated in the case of $\epsilon > 0$. One can show that proposition \ref{pro:macr-superp-unstable} can be generalised to $\epsilon>0$ under the assumption that the measurement on one group does not influences the statistics of a following measurement on an other group. The very same assumption is also used to admit an important and hence undesired role for the value of $\epsilon$, as noted in  \cite{frowis_measures_2012}. A detailed discussion and the proof can be found in \ref{sec:fagil-macr-superp}.

We conclude from the results in this section that for large system sizes it becomes infeasible to experimentally verify the generation of a coherent macroscopic superpositions in a strict sense. However, we stress that if we relax the definition of a macroscopic superposition, we can find certifiable states. For example, one can construct a so-called logical GHZ state, where the physical qubits are replaced by logical qubits that are groups of $m$ physical qubits. Define two logical states $\left| 0_L \right\rangle , \left| 1_L \right\rangle $ on each of the group. The logical GHZ state is then $\left| 0_L \right\rangle^{\otimes n} + \left| 1_L \right\rangle^{\otimes n} $, where $n$ is the number of logical groups. If $m = \log n$, the effective size equals $n$ and does not scale linearly with the system size $N = n \log n$. And with a suitable choice of the encoding \cite{frowis_stable_2011}, one can show --among other markers of stability-- the certifiability of the logical GHZ state for $m=\log n$ \footnote{The choice of the encoding made in  \cite{frowis_stable_2011} is $| 0_L \rangle  \propto \left| 0 \right\rangle ^{\otimes m} + \left| 1 \right\rangle ^{\otimes m}$ and $| 1_L \rangle  \propto \left| 0 \right\rangle ^{\otimes m} - \left| 1 \right\rangle ^{\otimes m}$. First, one shows that $\left| 0_L \right\rangle^{\otimes n} + \left| 1_L \right\rangle^{\otimes n} $ is distinguishable from $\left| 0_L \right\rangle^{\otimes n} - \left| 1_L \right\rangle^{\otimes n} $. The certifiability with respect to the remaining Hilbert space is shown by constructing a Hamiltonian as the negative sum of the stabilisers of the logical GHZ state that are at most $m$-local. The number of those stabilisers is $N-1$. Therefore, the ground state space is two-fold degenerated with the two states above as the ground states. Invoking proposition \ref{pro:degenerate-ground-states} of section \ref{sec:inst-with-resp}, one can show that the  certifiability scales at least with $p^m/poly(N) = n^{\log p}/poly(N)$, which is polynomial in $N$.}.

Furthermore, there is hope to handle general ``macroscopic quantum states'' (which is a more general concept than ``macroscopic superpositions'', based on the so-called Quantum Fisher information \cite{helstrom_quantum_1976,braunstein_statistical_1994,holevo_probabilistic_2011}). In  \cite{frowis_measures_2012}, it is argued that there exist quantum states that behave ``macroscopically'' in a similar sense as macroscopic superpositions. However, these quantum states do not exhibit this structure of being a superposition of two ``classical'' states. It is important to note that Dicke states $\left| N,k \right\rangle $ (see Example \ref{exp:DickeStates}) with $k = N/2$ are such general macroscopic quantum states and these states are asymptotically certifiable as shown in Example 3.

In short, to identify stable macroscopic quantum states, proposition \ref{pro:macr-superp-unstable} suggests to either use a sophisticated encoding or to seek for the generation of general macroscopic quantum states.

                                                                                                                                                                                           \section{The confusability index}
 \label{sec:inst-with-resp}

   In the example of the GHZ state, we found that small amounts of noise render  $| \mathrm{GHZ}  \rangle$ indistinguishable from the orthogonal state $ | \mathrm{GHZ}^{\bot}  \rangle$. We say that these two states are (asymptotically) \emph{confusable} under white noise. This raises the question whether the GHZ state is confusable with other states except $ | \mathrm{GHZ}^{\bot}  \rangle$. More generally, we may ask if there exist large sets of mutually confusable states i.e.~are there sets $\{|\psi_1\rangle, \dots, |\psi_n\rangle\}$ of mutually orthogonal states, with $n>2$, such that $D({\cal E}(\psi_i), {\cal E}(\psi_j))$ tends to zero faster-than-polynomially for every $i\neq j$? Such states  would  display a fragility under noise which is  stronger than seen in the case of the GHZ state, since it would be experimentally infeasible to distinguish between any of the $n$ orthogonal  states $|\psi_i\rangle$ in the presence of small amounts of noise. What is more, using an argument similar to the one in section \ref{sec:motiv_defn}, the trace distance between ${\cal E}(\psi_i)$ and any incoherent mixture from the family
   \begin{equation}
{\cal E} \left( \sum_{j=1}^n p_j \left| \phi_j \right\rangle\!\left\langle \phi_j\right|\right) \label{eq:6}
\end{equation}
would tend to zero faster-than-polynomially as well, where $\{p_i\}$ is an arbitrary probability distribution. This holds in particular for the equal mixture where $p_j=1/n$ for all $j$. The latter is a mixed state with entropy $n\log 2$. In this sense, if $n$ grows, the coherence of each of the superpositions $|\psi_i\rangle$ is increasingly fragile.

  In light of the above, we define the \emph{confusability index} of a state $|\psi\rangle$ in the presence of white noise to be the largest integer $n$ such that $|\psi\rangle$ belongs to a set of $n$ mutually orthogonal and mutually confusable states. For example, the GHZ state has confusability index at least 2. Below we argue that it is in fact precisely 2. Before doing so, we give two  examples of states with confusability index strictly larger than 2. These examples are constructed as generalisations of the GHZ state. The first example deals with states with a constant confusability index (i.e.~independent of $N$), the second example features states with superpolynomially large confusability index. It remains an open question whether there exists an $N$- qubit state with confusability index equal to $2^N$. This would amount to  an entire Hilbert space basis (i.e., $2^N$ orthogonal states) such that all basis vectors are mutually confusable.

 \begin{Exp}\label{ex:logGHZ}
Consider system sizes where $N/m \in \mathbbm{N}$ for a fixed $m\in \mathbbm{N}$. Form groups of $m$ qubits and take $\left| i \right\rangle $ ($i = 1,\dots, 2^m$) as the computational basis on this group. Define $\left| \psi_i \right\rangle = \left| i \right\rangle ^{\otimes N/m}$ and superpose these states equally weighted but with different phases such that we get $2^m$ orthogonal states of the form
                                                                                                                                                                                             \begin{equation}
\left| \Psi_k \right\rangle = \frac{1}{\sqrt{2^m}} \sum_{j = 1}^{2^m} e^{i \pi jk/2^{m}} \left| \psi_{j} \right\rangle,\label{eq:21}
\end{equation}
$k = 1,\dots,2^m$. Then, these states are mutually confusable as long as $m = O(1)$, because the coherence terms $\left| \psi_j \rangle\!\langle \psi_k\right| $ vanish exponentially fast with $N/m$ in the presence of white noise. This can be proved using the techniques of the proof of proposition \ref{pro:macr-superp-unstable}. In particular, under the action of small amounts of noise, all of these $2^m$ states are indistinguishable from the mixture $2^{-m}\sum_k \left| \Psi_k \rangle\!\langle \Psi_k\right| $.\end{Exp}

\begin{Exp}\label{ex:prodGHZ}
Define \be | \mathrm{GHZ}_n^{k} \rangle = \frac{1}{\sqrt{2}} (  | 0  \rangle ^{\otimes n} + (-1)^k  | 1  \rangle ^{\otimes n} )\ee and then consider product states of the form
 \begin{equation}
  | \Psi_{k_1,\dots,k_m}  \rangle =  | \mathrm{GHZ}_{N/m}^{k_{1}}  \rangle \otimes\dots \otimes | \mathrm{GHZ}_{N/m}^{k_m} \rangle\label{eq:22}
 \end{equation}
with $k_{l} \in  \{ 0,1  \}$.  If $m=O(N^{\alpha})$ with $0 \leq \alpha <1$, then the states $ | \Psi_{k_1,\dots,k_m}  \rangle$ are mutually confusable. This implies that each of these states has confusability index  $2^m$. In particular, in the presence of small amounts of white noise, each of these states becomes indistinguishable from the mixture \be 2^{-m} (  | 0 \rangle\!\langle 0 | ^{\otimes N/m} + | 1 \rangle\!\langle 1 |^{ \otimes N/m} )^{\otimes m} .\ee This mixture has almost-linear entropy $O(N^{\alpha})$ in the number of qubits $N$ and is thus highly incoherent.
                                                                                                  \end{Exp}
We have thus shown that there exist states with very large confusability index. Next we provide a technique to upper bound the confusability index of a state. The idea is to consider gapped local Hamiltonians which have the state as a---possibly degenerate---ground state:

\begin{Pro}
\label{pro:degenerate-ground-states}
  Let $\left| \psi \right\rangle \in  \mathcal{H}$ be the ground state of a $k$-local Hamiltonian $H$ with $\|H\|\leq $ poly$(N)$ and with ground state energy $E_0$. Suppose that the the ground space is $n$-fold degenerate and suppose that the first exited energy $E_1$ satisfies $|E_1-E_0|\geq $ 1/poly$(N)$. Then the confusability index of $|\psi\rangle$ is at most $n$.
 \end{Pro}
 The proof is obtained with similar techniques to the proof of proposition \ref{pro:unique-ground-states} and is omitted. The idea is to show that any state orthogonal to $|\psi\rangle$ with energy at least $E_1$ cannot be confusable with $|\psi\rangle$ since a measurement of the energy will be able to distinguish ${\cal E}(\psi)$ from such a state. Therefore, the only states which may be confusable with $|\psi\rangle$ are states in the ground space. Since the degeneracy is $n$, the confusability index of $|\psi\rangle$ is at most $n$.
 Proposition \ref{pro:degenerate-ground-states} can be use to show that the confusability index of  the GHZ state is precisely 2. To this end, consider $-J_z^2$, which is a local operator with a 2-fold ground eigenspace spanned by $|  \mathrm{GHZ}\rangle $ and $|  \mathrm{GHZ}^{\bot}\rangle $. Furthermore the difference between the ground state energy and the first excited energy is constant.

 Finally,  note that it may happen that a quantum state $ | \psi  \rangle $ is not confusable with $ | \phi_1  \rangle $ but $ | \phi_2  \rangle $ but still confusable with $1/\sqrt{2}  (  | \phi_1  \rangle +  | \phi_2  \rangle   )$. For example, define
  \be  | \psi  \rangle &=& \frac{1}{2}(  |00  \rangle^{\otimes \frac{N}{2}} +  |01  \rangle^{\otimes \frac{N}{2}}+  |10  \rangle^{\otimes \frac{N}{2}}+  |11  \rangle^{\otimes \frac{N}{2}} )\nonumber\\                                                                   | \phi_1  \rangle &=& \frac{1}{\sqrt{2}} (  |00  \rangle^{\otimes \frac{N}{2}} -  | 11  \rangle^{\otimes \frac{N}{2}}  ) \\  | \phi_2  \rangle &=& \frac{1}{\sqrt{2}}  (  |01  \rangle^{\otimes \frac{N}{2}} -  | 10 \rangle^{\otimes \frac{N}{2}}  ).\nonumber\ee
  Then $ | \psi  \rangle $ is confusable with $1/\sqrt{2}(  | \phi_1  \rangle +  | \phi_2  \rangle )$. However $ | \psi  \rangle $ is not confusable with $ | \phi_1  \rangle $ but $ | \phi_2  \rangle $ because these states can be efficiently distinguished by the local observable $\sum_{i = 1, 3,\dots} \sigma_z^{(i)} \sigma_z^{(i+1)}$.

  Conversely,  if a quantum state $| \psi\rangle $ is confusable with two  states $ | \xi_1  \rangle $ and $ | \xi_2  \rangle $,  this does not necessarily mean that $ | \psi  \rangle $ is confusable with $  (  | \xi_1 \rangle +  | \xi_2  \rangle   )/\sqrt{2}$. For example, the state $\left| \psi \right\rangle $ in the example above is confusable with $| \xi_1 \rangle  = (\left(  \left| \phi_1 \right\rangle + \left| \phi_2 \right\rangle \right)/\sqrt{2} $ and $| \xi_2 \rangle  =  \left(  \left| \phi_1 \right\rangle - \left| \phi_2 \right\rangle \right)/\sqrt{2}$, but not with  $( | \xi_1 \rangle + | \xi_2 \rangle)/\sqrt{2} = \left| \phi_1 \right\rangle $.

  \section{$\epsilon$-extension of fragility criterion}
  \label{sec:stab-neighb-inst}
  It is straightforward to extend our fragility criterion to a situation where one considers not only an initial pure state, but a certain neighbourhood of the state. This is similar in spirit to the study of smoothed entanglement measures \cite{mora_epsilon-measures_2008}. To this aim, one considers an $\epsilon$-ball around the initial state $|\psi\rangle$, i.e., the set of all density operators with trace distance $< \epsilon$ to $|\psi\rangle\langle\psi|$. More precisely, one identifies $| \psi \rangle $ with its $\epsilon$-ball environment. This is physically motivated by the fact that a given measurement apparatus with a certain accuracy and a restricted number of measurements has some finite accuracy $\delta$. Even without taking noise or decoherence into account, such a measurement will not be able to distinguish states within such an $\epsilon$ ball as long as $\epsilon < \delta$, so that in such a situation all states in the $\epsilon$-ball are indistinguishable. Similarly, for any state $|\phi\rangle$ orthogonal to $|\psi\rangle$ one considers an $\epsilon$-ball around it.

The state $|\psi\rangle$ is then confusable with $|\phi\rangle$ in this smooth sense if---after the action of noise ${\cal E}$---there exist states in the noisy $\epsilon$-balls around $|\phi\rangle$ and $|\psi\rangle$ which are $\delta$-close to each other. Notice that completely positive maps are contractive, i.e., the radius of the $\epsilon$-balls in fact shrinks.
This implies that if a pure state is certifiable with $C>\delta$ and $C>2\epsilon$, the $\epsilon$-ball around $| \psi \rangle $ is also certifiable in the sense that we can experimentally exclude the $\epsilon$-ball environment of any initially orthogonal state. On the other hand, states within the immediate neighbourhood of an incertifiable state (i.e., the $\epsilon$-ball of $| \psi \rangle $ contains an incertifiable state) are then considered as fragile in this smooth sense.

  \section{Summary and conclusions}
\label{sec:conclusion-summary}
We have introduced a fragility criterion that is based on the distinguishability of orthogonal states in the presence of small amounts of noise $\mathcal{E}$ and finite measurement precision $\delta$. The basic idea is that if a quantum state $\left| \psi \right\rangle $ is indistinguishable from an orthogonal state $| \psi^{\bot} \rangle $ in this sense, it is also indistinguishable from any incoherent mixture $a \left| \psi \rangle\!\langle \psi\right| + (1-a) | \psi^{\bot} \rangle\!\langle \psi^{\bot}| $, that is, the coherence of $\left| \psi \right\rangle $ is not fully certifiable in the experiment. This general principle was pursued for local depolarisation channels, where we were in particular interested in the asymptotic behaviour of states defined on large systems. %
Our investigation is similar in spirit to other proposals aimed at characterizing the instability or fragility of certain quantum states. In particular, our work connects to the program of explaining the appearance of the classical (macroscopic) world out of quantum mechanics with the help of decoherence. A notable feature of our contribution is that the fragility criterion proposed in this paper gives rise to qualitative differences between various classes of many-body qubit states (e.g. product states versus GHZ states). This is in contrast to other quantities like fidelity or entropy, which behave similarly for many state families, especially when considering noise processes which map any state to a full rank density operator (such as local white noise, as considered in this paper). For example, the entropy scaling of a product state and a GHZ state under the action of local white noise will be qualitatively very similar, in both cases growing linearly with the system size. 

 We found that product states are asymptotically certifiably, that is, one can distinguish them from noisy orthogonal states with a finite measurement precision and in the presence of a small amount of white noise. More generally, we proved that all unique ground states of local, gapped Hamiltonians are certifiable. This includes interesting classes of states like Dicke states and graph states with a bounded degree, such as cluster states.

 In contrast, the entire class of macroscopic superpositions was shown to be asymptotically incertifiable. If $\left| \psi \right\rangle \propto \left| \psi_0 \right\rangle + \left| \psi_1 \right\rangle $ is macroscopic (where macroscopicity is defined relative to several measures currently used the literature), the application of the individual depolarisation channels render this quantum state indistinguishable from $\left| \psi_0 \rangle\!\langle \psi_0\right| + \left| \psi_1 \rangle\!\langle \psi_1\right| $. Since the states $\left| \psi_0 \right\rangle$ and $\left| \psi_1 \right\rangle $ do not have any ``macroscopic quantum properties'', the ``macroscopicity'' of $\left| \psi \right\rangle $ can not be certified in a realistic scenario. As for experiments, this implies that for large $N$ it is in practice impossible to certify the generation of, for instance, the GHZ state. However, as discussed in section \ref{sec:asymptotic_certifiability}, for any fixed system size, the certifiability of any Schr\"odinger cat state can be increased by using a sophisticated encoding to protect (to a certain extent) the quantum information of the state. This is in the same spirit as error correction for quantum computation. In addition, more general macroscopic quantum states that are not a superposition of exactly two ``classical'' states may in principle still exhibit stability i.e.~such states are a priori not covered by our results.

Finally we constructed GHZ-like examples of state families that are confusable with (super-polynomial) many orthogonal states. 
 \ack 
 This research was funded by the Austrian Science Fund (FWF): P24273-N16, F40-FoQus F4012-N16 and J3462.

\appendix
\section{Basic terminology}
\textit{Order---} Throughout this paper, the order of a quantity is always understood as with respect to the system size $N$. We use the following definitions: $f(N) = O(N^x) :\Leftrightarrow \lim_{N\rightarrow \infty}f(N)/N^x$ exists and is strictly greater than zero; $f(N) = o(N^x) :\Leftrightarrow \lim_{N\rightarrow \infty}f(N)/N^x$ exists and equals zero.

\textit{Trace norm---} The trace norm $\lVert \cdot \rVert_1$ of a linear operator $A$ is defined as $\lVert A \rVert_1 \mathrel{\mathop:}= \mathrm{Tr}\left| A \right| \equiv \mathrm{Tr}\sqrt{A A^{\dagger}} $ or, equivalently, as the sum of its singular values. If $A$ is hermitian, the trace norm equals the sum of the moduli of the eigenvalues of $A$. Therefore, the trace norm of any quantum state is unity. The trace norm is multiplicative, that is, $\lVert A\otimes B \rVert_1 = \lVert A \rVert_1 \lVert B \rVert_1$.

\textit{Pauli operator---} The three Pauli operators $\left( \sigma_x, \sigma_y,\sigma_z \right) \equiv \left( \sigma_1, \sigma_2, \sigma_3 \right): \mathbbm{C}^2 \rightarrow \mathbbm{C}^2$ are linear, hermitian and traceless operators with spectra $\left\{ -1,1 \right\}$ that fulfil the commutation relations $\left[ \sigma_j, \sigma_k \right] = 2i \epsilon_{jkl} \sigma_l$ and form (together with the identity operator $\mathbbm{1}\equiv \sigma_0$) an orthogonal basis of the space of complex two-by-two matrices.

\textit{Pauli basis---}In the space of all complex matrices that map elements of $\mathbbm{C}^{2 \otimes N}$ onto $\mathbbm{C}^{2 \otimes N}$, one can use the Pauli basis that consists of tensor products of the Pauli operators and the identity, that is, $\sigma_{i_1} \otimes \dots \otimes \sigma_{i_N}$ with $i_k \in \left\{ 0,\dots,3 \right\}$ for all $k$.

\textit{Weight---} Given an operator $\sigma = \sigma_{i_1} \otimes \dots \otimes \sigma_{i_{N}}$ of the Pauli basis, its weight $w(\sigma)$ is defined as the number of matrices in the tensor product which are different from the identity. More formally,  \be w(\sigma)= |\{k\in\{1, \dots, N\}: i_k\neq 0\}|.\ee

\section{Lemma for proof of proposition \ref{pro:macr-superp-unstable}}
                                                                                                  \label{sec:lemma-proof-prop}
The following lemma is used to prove proposition \ref{pro:macr-superp-unstable}.
\begin{Lem}
 \label{sec:stab-stat-with--Lem2}
For every $p \in (0,1]$ and $m \in \mathbbm{N}$, the linear map
\begin{equation}
\begin{split}
   \tilde{\mathcal{E}}: \mathcal{D}(\mathbbm{C}^{\otimes m})&\rightarrow \mathcal{D}(\mathbbm{C}^{\otimes m})\\
\rho &\mapsto \frac{1}{1-(1-p)^m}\sum_{k=0}^{m-1}p^{m-k}(1-p)^k \sum_{i_1 < \dots < i_k} \mathrm{Tr}_{i_1,\dots,i_k} \rho \otimes \frac{\mathbbm{1}^{\otimes k}}{2^k}
 \end{split}
 \label{eq:92}
 \end{equation}
 is a cp map.
 \end{Lem}
 \begin{proof}
 According to  \cite[Thm 8.1]{nielsen_quantum_2010}, it suffices to show that one can write  (\ref{eq:92}) as an operator sum, that is, $\tilde{\mathcal{E}}(\rho) = \sum_j K_j \rho K_j^{\dagger}$ with $\sum_j K_j^{\dagger} K_j = \mathbbm{1}^{\otimes m}$. We prove in the following that this is possible with the choice $K_\mathbf{j} = \sqrt{c_\mathbf{j}}\sigma_{j_1}^{(i_1)} \dots \sigma_{j_k}^{(i_k)}$ where \begin{equation}
 \label{eq:95}
 c_\mathbf{j} = \frac{1}{4^k\left[ 1-(1-p)^m \right]} p^{m-k}(1-p)^k.
 \end{equation} Note that $\mathbf{j}$ is a multi-index defined as $\mathbf{j}(k) \equiv \mathbf{j} = \left( j_1,i_1,\dots, j_k, i_k \right)$ with $j_l \in \left\{ 0,\dots,3 \right\} $ and $1 \leq i_1 < \dots < i_k \leq m$. The index represents a specific choice of $k$ qubits and the respective Pauli operators acting on them.

If $k=1$ in  (\ref{eq:92}), that is, if a single particle is traced out, one can easily verify that $\mathrm{Tr}_i \rho \otimes \mathbbm{1}= \frac{1}{2} \sum_{j=0}^3\sigma_j^{(i)} \rho \sigma_j^{(i)}$. If more particles are traced out, we observe that
\begin{equation}
 \label{eq:93}
   \begin{split}
  &\mathrm{Tr}_{i_1,\dots,i_k} \rho \otimes \mathbbm{1}^{\otimes k} =
   \mathrm{Tr}_{i_1}\left\{ \mathrm{Tr}_{i_2}\left[ \dots  \mathrm{Tr}_{i_k}(\rho) \dots \right] \right\}\otimes  \mathbbm{1}^{\otimes k} =\\ &
\sum_{j_1=0}^3\sigma_{j_1}^{(i_1)}\left\{\sum_{j_2=0}^3\sigma_{j_2}^{(i_2)}\left[ \dots \sum_{j_k=0}^3\sigma_{j_k}^{(i_k)}\rho\sigma_{j_k}^{(i_k)}  \dots \right]\sigma_{j_2}^{(i_2)} \right\}\sigma_{j_1}^{(i_1)} \\ &  =\frac{1}{2^k}  \sum_{j_1,\dots, j_k = 0}^3 \sigma_{j_1}^{i_1}\dots  \sigma_{j_k}^{i_k} \rho \sigma_{j_1}^{i_1}\dots  \sigma_{j_k}^{i_k}  \\ & =\frac{1}{2^k}
 \sum_{j_1,\dots, j_k = 0}^3 \frac{1}{c_j}K_{\mathbf{j}} \rho K_{\mathbf{j}}^{\dagger}.  \end{split}  \end{equation}
Therefore, we have shown that $\tilde{\mathcal{E}}(\rho) = \sum_\mathbf{j} K_\mathbf{j} \rho K_\mathbf{j}^{\dagger}$. We are left to prove $\sum_\mathbf{j} K_\mathbf{j}^{\dagger} K_\mathbf{j} = \mathbbm{1}^{\otimes m}$. First, note that the partial sum $\sum_{\mathbf{j}: k= \mathrm{const}}K_\mathbf{j}^{\dagger} K_\mathbf{j}$ consists of $4^k$ addends which are equal to $c_\mathbf{j} \mathbbm{1}^{\otimes m}$. Therefore  \begin{equation} \label{eq:99}
 \begin{split}
&\sum_\mathbf{j} K_{\mathbf{j}}^{\dagger} K_{\mathbf{j}} = \sum_{k^{\prime}=0}^{m-1}\sum_{\mathbf{j}:  k=k^{\prime} }K_{\mathbf{j}(k)}^{\dagger} K_{\mathbf{j}(k)} = \sum_{k=0}^{m-1}4^{k}c_{\mathbf{j}}  \mathbbm{1}^{\otimes m} =\\ & \frac{1}{ 1-(1-p)^m }\sum_{k=0}^{m-1} p^{m-k}(1-p)^k \mathbbm{1}^{\otimes m}= \mathbbm{1}^{\otimes m}.
 \end{split}\end{equation}
For the last equality, we use the binomial theorem. \end{proof}

  \section{Fragility of macroscopic superpositions with finite success probability}
\label{sec:fagil-macr-superp}
 In this section, we extend proposition \ref{pro:macr-superp-unstable} of section \ref{sec:macr-superp-are} to the case where the success probability $P=1-\epsilon$ is not unity, that is, the measurement of one subgroup does not give us full certainty whether $\left| \psi_0 \right\rangle $ or $\left| \psi_1 \right\rangle $ was present before. For $\epsilon>0$, the situation is more complex than for $\epsilon = 0$.

First of all, it is not clear how the threshold of $\epsilon$ should be set. In  \cite{frowis_measures_2012}, an example is discussed, where $N_{\mathrm{eff}} = 2 \epsilon (N+1)$, which admits a prominent and hence undesired role to $\epsilon$ for $N_{\mathrm{eff}}$. There, it has been argued that the additional requirement that the measurements on any group does not influence the measurement outcomes on any other group resolves this problem. Mathematically, let $A^{(i)}$ and $A^{(j)}$ be two measurements that act on any two different groups $i$ and $j$, respectively, and distinguish optimally between $\left| \psi_0 \right\rangle $ and $\left| \psi_1 \right\rangle $. We then demand that $\langle A^{(i)} A^{(j)} \rangle_{\psi_k} = \langle A^{(i)} \rangle_{\psi_k} \langle A^{(j)} \rangle_{\psi_k} + o(1)$ for $k=0,1$ and for almost all groups $i,j$. It turns out that this requirement is also useful for the extension of proposition \ref{pro:macr-superp-unstable} if $\epsilon>0$.
                                
\begin{Pro}
\label{pro:macr-superp-are-2}
Let $\left| \psi \right\rangle  = 1/\sqrt{2} \left( \left| \psi_0 \right\rangle + \left| \psi_1 \right\rangle  \right) \in \mathcal{H}$ be a macroscopic superposition (with respect to  \cite{korsbakken_measurement-based_2007}) with $N_{\mathrm{eff}} = O(N)$. Suppose a success probability to distinguish $\left| \psi_0 \right\rangle $ from $\left| \psi_1 \right\rangle $ of $P = 1-\epsilon$ with $0<\epsilon<1/2$ for any group. Furthermore, the outcome for the optimal measurement of any group shall not influence the measurement of any other group. Then, for $p \in \left( 0,1 \right)$, there exist $q,\epsilon \in \left( 0,1 \right)$ such that
\begin{equation}
\label{eq:24}
C(\psi) \leq \left[ q + \epsilon (1-q) \right]^{N_{\mathrm{eff}}}.\end{equation}
 \end{Pro}
The proof uses some expressions that were introduced in the proof of proposition \ref{pro:macr-superp-unstable}. \begin{proof} 
 Denote the optimal measurement operator on group $i$ by $A^{(i)}$ with spectrum $\left\{ -1,1 \right\}$. The success probability to distinguish $\left| \psi_0 \right\rangle $ from $\left| \psi_1 \right\rangle $ is then given by (\cite{korsbakken_measurement-based_2007}) $P = 1/2 + 1/4 \left| \langle A^{(i)} \rangle_{\psi_0} - \langle A^{(i)} \rangle_{\psi_1} \right|$. Without loss of generality, we assume that $ \langle A^{(i)} \rangle_{\psi_0} = 1-\epsilon$ and $ \langle A^{(i)} \rangle_{\psi_1} = -1+\epsilon$ for all $i$. The projection onto the measurement outcome $l$ is denoted by $\Pi^{(i)}_l$. Then, one finds the probabilities $p_{0,l}^{(i)} = \lVert \Pi_l^{(i)} \left| \psi_0 \right\rangle  \rVert \in \left\{  1- \epsilon/2, \epsilon/2\right\}$ for $l \in \left\{ +1,-1 \right\}$, respectively, and $p_{1,l}^{(i)} \in \left\{\epsilon/2,  1- \epsilon/2\right\}$ for $\left| \psi_1 \right\rangle $. Note that the independence of the measurements lead to $p_{0,l{l^{\prime}}}^{(i,j)} =  \lVert \Pi_l^{(i)}\Pi_{l^{\prime}}^{(j)} \left| \psi_0 \right\rangle  \rVert  = p_{0,l}^{(i)}p_{0,{l^{\prime}}}^{(j)}$ and similarly for $\left| \psi_1 \right\rangle $.

Continuing with  (\ref{eq:19}) from the proof of proposition \ref{pro:macr-superp-unstable}, one has\begin{equation}
                                                                                                      \label{eq:25}
     \begin{split}
       &C(\psi) \leq \lVert \mathcal{E}_g^{\otimes N_{\mathrm{eff}}}(\left| \psi_0
       \rangle\!\langle \psi_1\right| ) \rVert_1 \leq\\ & \sum_{k=0}^{N_{\mathrm{eff}}}
         q^{N_{\mathrm{eff}}-k}(1-q)^k\sum_{j_1 < \ldots < j_k}\lVert\mathrm{Tr}_{j_1,\ldots,j_k}\left| \psi_0 \rangle\!\langle  \psi_1\right|  \rVert_1,  \end{split}
 \end{equation}
 where we now trace out entire groups denoted by $j_l$. Let us consider a single term $ \lVert\mathrm{Tr}_{j_1,\ldots,j_k}\left| \psi_0 \rangle\!\langle \psi_1\right|  \rVert_1$. We insert the twice the identity $\left( \Pi_{+1}^{(j_1)} + \Pi_{-1}^{(j_1)} \right)\dots \left( \Pi_{+1}^{(j_k)} + \Pi_{-1}^{(j_k)} \right)$ within the trace and get
 
\begin{equation}
 \label{eq:26}
 \begin{split}
 & \lVert\mathrm{Tr}_{j_1,\ldots,j_k}\left| \psi_0 \rangle\!\langle
  \psi_1\right| \rVert_1 =\\ &
  \lVert\mathrm{Tr}_{j_1,\ldots,j_k}\left(\sum_{l_1,\dots,l_k =
  \pm 1} \Pi_{l_1}^{(j_1)} \dots \Pi_{l_k}^{(j_k)} \left|
   \psi_0 \rangle\!\langle
  \psi_1\right|\sum_{l^{\prime}_1,\dots,l^{\prime}_k = \pm 1}
  \Pi_{l^{\prime}_1}^{(j_1)} \dots
    \Pi_{l^{\prime}_k}^{(j_k)}\right) \rVert_1 = \\ &
   \lVert\mathrm{Tr}_{j_1,\ldots,j_k}\left(\sum_{l_1,\dots,l_k =
   \pm 1} \Pi_{l_1}^{(j_1)} \dots \Pi_{l_k}^{(j_k)} \left|
   \psi_0 \rangle\!\langle\psi_1\right|\Pi_{l_1}^{(j_1)} \dots
    \Pi_{l_k}^{(j_k)}\right) \rVert_1 \leq\\ &
    \sum_{l_1,\dots,l_k = \pm 1}  \lVert\mathrm{Tr}_{j_1,\ldots,j_k}   (1-\epsilon/2)^L (\epsilon/2)^{k-L} \left| \psi_0^{l_1,\dots,l_k} \right\rangle\! \left\langle  \psi_1^{l_1,\dots,l_k}\right| (1-\epsilon/2)^{k-L} (\epsilon/2)^{L}\rVert_1 = \\ & \left[ \left( 1-\epsilon/2 \right) \epsilon \right]^k \equiv \epsilon^k,\end{split} \end{equation}
    where we use $\Pi_l^{(i)}\Pi_{l^{\prime}}^{(i)}  = \epsilon_{ll^{\prime}}\Pi_l^{(i)}$. $|  \psi_i^{l_1,\dots,l_k}\rangle $ is the normalised state after the application of the projections onto $\left| \psi_i \right\rangle $ and $L = \sum_{l_1,\dots,l_k = \pm 1} \theta(l_k)$ the number of positive eigenvalues. With this and the binomial theorem, one can estimate  (\ref{eq:25}) further to  \begin{equation}
 \label{eq:27}
 C(\psi) \leq\sum_{k=0}^{N_{\mathrm{eff}}}
  q^{N_{\mathrm{eff}}-k}(1-q)^k\sum_{j_1 < \ldots < j_k} \epsilon^k = 
  \sum_{k=0}^{N_{\mathrm{eff}}} \binom{N_{\mathrm{eff}}}{k}
  q^{N_{\mathrm{eff}}-k}(1-q)^k\epsilon^k =  
  \left[ q + \epsilon \left(  1-q \right) \right]^{N_{\mathrm{eff}}}.  
  \end{equation}
 \end{proof}
   
\section*{References}
\label{sec:references}

\bibliographystyle{iopart-numMOD}
\bibliography{References}
\end{document}